\documentclass[twoside,11pt]{article}

\usepackage{amsmath}
\usepackage{wrapfig}
\usepackage{clrscode}
\usepackage{supertech}

\title{A New Approach to Incremental Cycle Detection\\and Related Problems}

\author{
  Michael A. Bender\\\textit{Department of Computer Science}\\\textit{Stony Brook
    University}
  \and Jeremy T. Fineman\\\textit{Department of Computer Science}\\\textit{Georgetown University}\vspace{1em}
  \and Seth Gilbert\\\textit{Department of Computer Science}\\\textit{National University of Singapore}
  \and Robert E. Tarjan\\\textit{HP}\\{\normalsize and}\\\textit{Department of Computer Science}\\\textit{Princeton
  University}
}
\date{}

\begin{document}

\footnotenonumber{
  This work was supported in part by the National Science Foundation,
  under grants 
  CCF-0621439/0621425, 
  CCF-0540897/05414009, 
  CCF-0634793/0632838, 
  and CNS-0627645 
  for M. A. Bender, 
  CCF-0621511, CNS-0615215, CCF-0541209, and NSF/CRA
  sponsored CIFellows program 
  for J. T. Fineman, and CCF-0830676 
  and CCF-0832797 for R. E. Tarjan.  
  The information contained herein does
  not necessarily reflect the opinion or policy of the federal
  government and no official endorsement should be inferred.}

\maketitle

\begin{abstract} 
  We consider the problem of detecting a cycle in a directed graph
  that grows by arc insertions, and the related problems of
  maintaining a topological order and the strong components of such a
  graph.  For these problems we give two algorithms, one suited to
  sparse graphs, the other to dense graphs.  The former takes
  $O(\min\{m^{1/2}, n^{2/3}\}m)$ time to insert $m$ arcs into an $n$-vertex
  graph; the latter takes $O(n^2 \log n)$ time.  Our sparse algorithm is
  considerably simpler than a previous $O(m^{3/2})$-time algorithm; it
  is also faster on graphs of sufficient density.  The time bound of
  our dense algorithm beats the previously best time bound of
  $O(n^{5/2})$ for dense graphs.  Our algorithms rely for their
  efficiency on topologically ordered vertex numberings; bounds on the
  size of the numbers give bounds on running time.
\end{abstract}

\secput{intro}{Introduction}

Perhaps the most basic algorithmic problem on directed graphs is cycle
detection.  We consider an incremental version of this problem: given
an initially empty graph that grows by on-line arc insertions, report
the first insertion that creates a cycle.  We also consider two
related problems, that of maintaining a topological order of an
acyclic graph as arcs are inserted, and maintaining the strong
components of such a graph.

We use the following terminology.  We order pairs lexicographically:
$a, b < c, d$ if and only if either $a < b$, or $a = b$ and $c < d$.
We denote a list by square brackets around its elements; ``[ ]''
denotes the empty list.  We denote list catenation by ``\&''.  In a
directed graph, we denote an arc from $v$ to $w$ by $(v, w)$.  We
disallow multiple arcs and loops (arcs of the form $(v, v)$).  We
assume that the set of vertices is fixed and known in advance,
although our results extend to handle on-line vertex insertions.  We
denote by $n$ and $m$ the number of vertices and arcs, respectively.
We assume that $m$ is known in advance; our results extend to handle
the alternative.  To simplify expressions for bounds we assume $n > 1$
and $m = \Omega(n)$; both are true if there are no isolated vertices.
A vertex $v$ is a \defn{predecessor} of $w$ if $(v, w)$ is an arc.
The \defn{size} $\id{size}(w)$ of a vertex $w$ is the number of
vertices $v$ such that there is a path from $v$ to $w$.  Two vertices,
two arcs, or a vertex and an arc are \defn{related} if they are on a
common path, \defn{mutually related} if they are on a common cycle
(not necessarily simple), and \defn{unrelated} if they are not on a
common path.  Relatedness is a symmetric relation.  The \defn{strong
  components} of a directed graph are the subgraphs induced by the
maximal subsets of mutually related vertices.

A \defn{dag} is a directed acyclic graph.  A \defn{weak topological
  order} $<$ of a dag is a partial order of the vertices such that if
$(v, w)$ is an arc, $v < w$; a \defn{topological order} of a dag is a
total order of the vertices that is a weak topological order.  A
\defn{weak topological numbering} of a dag is a numbering of the
vertices such that increasing numeric order is a weak topological
order; a \defn{topological numbering} of a dag is a numbering of the
vertices from 1 through n such that increasing numeric order is a
topological order.

There has been much recent work on incremental cycle detection,
topological ordering, and strong component
maintenance~\cite{AjwaniFr07,AjwaniFrMe06,AlpernHoRo90,HaeuplerKaMa08,Katriel04,KatrielBo05,KavithaMa07,LiuCh07,Marchetti-SpaccamelaNaRo96,PearceKe04,PearceKeHa03,HaeuplerKaMa?}. For a thorough discussion of this work
see~\cite{HaeuplerKaMa08,HaeuplerKaMa?}; here we discuss the
heretofore best results and others related to our work.  A classic
result of graph theory is that a directed graph is acyclic if and only
if it has a topological order~\cite{Szpilrajn30}; a more recent
generalization is that the strong components of a directed graph can
be ordered topologically (so that every arc lies within a component or
leads from a smaller component to a larger one)~\cite{HararyNoCa65}.
For static graphs, there are two $O(m)$-time algorithms to find a
cycle or a topological order: repeated deletion of vertices with no
predecessors~\cite{Knuth73,KnuthSz74} and depth-first
search~\cite{Tarjan72}: the reverse postorder~\cite{Tarjan74} defined
by such a search is a topological order if the graph is acyclic.
Depth-first search extends to find the strong components and a
topological order of them in $O(m)$ time~\cite{Tarjan72}

For incremental cycle detection, topological ordering, and strong
component maintenance, there are two known fastest algorithms, one
suited to sparse graphs, the other suited to dense graphs.  Both are
due to Haeupler et al.~\cite{HaeuplerKaMa08,HaeuplerKaMa?}.
Henceforth we denote the coauthors of these papers by HKMST.  The HKMST
sparse algorithm takes $O(m^{3/2})$ time for $m$ arc additions; the
HKMST dense algorithm takes $O(n^{5/2})$ time.  Both of these
algorithms use two-way search; each is a faster version of an older
algorithm.  These algorithms, and the older ones on which they are
based, bound the total running time by counting the number of arc
pairs or vertex pairs that become related as a result of arc
insertions.  The HKMST sparse algorithm uses a somewhat complicated
dynamic list data structure~\cite{DietzSl87,BenderCoDe02} to represent
a topological order, and it uses either linear-time selection or
random sampling to guide the searches.  There are examples on which
the algorithm takes $\Omega(nm^{1/2})$ time, so its time bound is
tight for sparse graphs.  The time bound of the HKMST dense algorithm
is not known to be tight, but there are examples on which it takes
$\Omega(n^2 2^{\sqrt{2 \lg n}})$~\cite{HaeuplerKaMa?}.

Our approach to incremental cycle detection and the related problems
is different.  We maintain a weak topological numbering and use it to
facilitate cycle detection.  Our algorithms pay for cycle-detecting
searches by increasing the numbers of appropriate vertices; a bound on
the numbers gives a bound on the running time.  One insight is that
the size function is a weak topological numbering.  Unfortunately,
maintaining this function as arcs are inserted seems to be expensive.
But we are able to maintain in $O(n^2\log n)$ time a weak topological
numbering that is a lower bound on size.  This gives an incremental
cycle detection algorithm with the same running time, substantially
improving the time bound of the HKMST dense algorithm.  Our algorithm
uses one-way rather than two-way search.  For sparse graphs, we use a
two-part numbering scheme. The first part is a scaled lower bound on
size, and the second part breaks ties.  This idea yields an algorithm
with a running time of $O(\min\{m^{1/2}, n^{2/3}\}m)$. Our algorithm is
substantially simpler than the HKMST sparse algorithm and
asymptotically faster on sufficiently dense graphs.  The $O(n^2\log
n)$ algorithm appeared previously in~\cite{BenderFiGi09}, but the
other algorithm is new to this paper. 

The remainder of our paper consists of four sections.  \secref{sparse}
describes the two versions of our cycle-detection algorithm for sparse
graphs.  \secref{dense} describes our cycle-detection algorithm for
dense graphs.  \secref{extensions} describes several simple extensions
of the algorithms.  \secref{components} extends the algorithms to
maintain the strong components of the graph as arcs are inserted
instead of stopping as soon as a cycle exists.  The extensions in
\secreftwo{extensions}{components} preserve the asymptotic time bounds
of the algorithms. \secref{conc} contains concluding remarks.

\secput{sparse}{A Two-Way-Search Algorithm for Sparse Graphs}

Our algorithm for sparse graphs uses two-way search.  Unlike the
entirely symmetric forward and backward searches in the HKMST sparse
algorithm, the two searches in our algorithm have different functions.
Also unlike the HKMST sparse algorithm, our algorithm avoids the use
of a dynamic list data structure, and it does not use selection or
random sampling: all of its data structures are simple, as is the
algorithm itself.

We use a two-part numbering scheme whose lexicographic order is
topological.  Specifically, we partition the vertices into levels.  We
maintain a weak topological numbering of the levels; within each
level, we give the vertices indices ordered topologically.  Each
backward search proceeds entirely within a level.  If the search takes
too long, we stop it and increase the level of a vertex.  This bounds
the backward search time.  Each forward search traverses only arcs
that lead to a lower level, and it increases the level of each vertex
visited.  An overall bound on such increases gives a bound on the time
of all the forward searches.  If the backward and forward searches do
not detect a cycle, we update vertex indices to restore topological
order.  To facilitate this, we make the searches depth-first.

Each vertex $v$ has a \defn{level} $k(v)$ and an \defn{index} $i(v)$.
Indices are distinct.  While the graph remains acyclic, lexicographic
order on level and index is a topological order.  That is, if $(x, y)$
is an arc, $k(x), i(x) < k(y), i(y)$.  Levels are positive integers and
indices are negative integers.  We make indices negative because newly
assigned indices must be smaller than old ones.  An alternative is to
maintain the negatives of the indices and reverse the sense of all
index comparisons.  Initially, each vertex $v$ has $k(v) = 1$ and $i(v)$ an
integer between $-n$ and $-1$ inclusive, distinct for each vertex.

In addition to levels and indices, we maintain a variable $\id{index}$
equal to the smallest index assigned so far.  To represent the graph,
we maintain for each vertex $v$ the set $\id{out}(v)$ of outgoing arcs
$(v, w)$ and the set $\id{in}(v)$ of incoming arcs $(u, v)$ such that
$k(u) = k(v)$.  Initially $\id{index} = -n$ and all incident arc sets
are empty.  Each backward search marks the vertices it visits.
Initially all vertices are unmarked.  To bound backward searches, we
count arc traversals.  Let $\Delta = \min\set{m^{1/2}, n^{2/3}}$.
Recall that we denote a list with square brackets ``[ ]'' around its
elements and list concatenation with ``\&.''

The algorithm for inserting a new arc $(v,w)$ consists of the
following steps:

\noindent \textbf{Step 1 (test order)}: If $k(v), i(v) < k(w),i(w)$ go to Step 5
(lexicographic order remains topological).

\noindent \textbf{Step 2 (search backward)}: Let $B \gets [\ ]$, $F
\gets [\ ]$, and $\id{arcs} \gets 0$, where $\gets$ denotes the
assignment operator.  Do $\proc{Bvisit}(v)$, where mutually recursive
procedures \proc{Bvisit} and \proc{Btraverse} are defined as follows:

\begin{codebox}
  \Procname{$\procdecl{Bvisit}(y)$}
  \li mark $y$
  \li \For $(x,y) \in \id{in}(y)$
  \li \Do $\proc{Btraverse}(x,y)$ 
      \End
  \li $B \gets B \& [y]$
\end{codebox}
\begin{codebox}
  \Procname{$\proc{Btraverse}(x,y)$}
  \li \If $x = w$
  \li \Then stop the algorithm and report the detection of a cycle.
      \End
  \li $\id{arcs} \gets \id{arcs} + 1$  
  \li \If $\id{arcs} \geq \Delta$ 
  \li \Then (The search ends, having traversed at least
  $\Delta$ arcs without reaching $w$.)
  \li       $k(w) \gets k(v) + 1$
  \li       $\id{in}(w) \gets \set{}$
  \li       $B \gets [\ ]$
  \li       unmark all marked vertices
  \li       go to Step 3 (aborting the backward search)
      \End
  \li \If $x$ is unmarked
  \li \Then $\proc{Bvisit}(x)$
      \End
\end{codebox}

If the search ends without detecting a cycle or traversing at least
$\Delta$ arcs, test whether $k(w) = k(v)$.  If so, go to Step 4; if
not, let $k(w) = k(v)$ and $\id{in}(w) = \set{}$.  

\noindent \textbf{Step 3 (forward search)}: Do $\proc{Fvisit}(w)$,
  where mutually recursive procedures \proc{Fvisit} and
  \proc{Ftraverse} are defined as follows:

\begin{codebox}
  \Procname{$\procdecl{Fvisit}(x)$}
  \li \For $(x,y) \in \id{out}(x)$
  \li \Do $\proc{Ftraverse}(x,y)$
      \End
  \li $F \gets [x] \& F$
\end{codebox}
\ 

\begin{codebox}
  \Procname{$\procdecl{Ftraverse}(x,y)$}
  \li \If $y=v$ or $y$ is in $B$
  \li \Then stop the algorithm and report the detection of a cycle
      \End
  \li \If $k(y) < k(w)$
  \li \Then $k(y) \gets k(w)$
  \li       $\id{in}(y) \gets \set{}$
  \li       $\proc{Fvisit}(y)$
      \End
  \li \Comment{Now, $k(y) \geq k(w)$}
  \li \If $k(y) = k(w)$
  \li \Then add $(x,y)$ to $\id{in}(y)$
      \End
\end{codebox}

\noindent \textbf{Step 4 (re-index)}: Let $L = B \& F$. While $L$ is
nonempty, let $\id{index} = \id{index} - 1$, delete the last vertex
$x$ on $L$, and let $i(x) = \id{index}$.  

\noindent \textbf{Step 5 (insert arc)}: Add $(v,w)$ to $\id{out}(v)$.
If $k(v) = k(w)$, add $(v,w)$ to $\id{in}(w)$.

\begin{theorem}
  If a new arc creates a cycle, the insertion algorithm stops and
  reports a cycle.  If not, lexicographic order on level and index is
  a topological order.\thmlabel{sparsecorrect}
\end{theorem}
\begin{proof}
  By inspection, the algorithm correctly maintains the incident arc sets
  and the value of $\id{index}$.  We prove the theorem by induction on
  the number of arc insertions. Initially, lexicographic order on
  level and index is a topological order since there are no arcs:
  \emph{any} total order is topological. Suppose the lexicographic
  order is topological just before the insertion of an arc $(v,w)$.
  If $k(v), i(v) < k(w),i(w)$ before the insertion, then lexicographic
  order on level and index remains topological.  Thus assume
  $k(v),i(v) > k(w),i(w)$.  

  Clearly, if the algorithm stops and reports a cycle, there is one.
  Suppose the insertion of $(v,w)$ creates a cycle.  Such a cycle
  consists of the arc $(v,w)$ and a pre-existing path from $w$ to $v$,
  along which levels are nondecreasing before he insertion.  If $v$
  and $w$ have the same level, then all vertices on the path have the
  same level, and either the backward search will traverse the entire
  path and report a cycle, or it will report a different cycle, or it
  will stop, $w$ will increase in level, and the algorithm will
  proceed to Step 3.  If $v$ has larger level than $w$, $w$ will
  increase in level in Step 2, and the algorithm will do Step 3.
  Suppose the algorithm does Step 3.  At the beginning of this step,
  vertex $w$ has maximum level on the cycle, and $B$ is the set of
  vertices from which $v$ is reachable by a path all of whose vertices
  have level $k(w)$.  (Either $k(w) = k(v)$, in which case $v$ is in
  $B$, or $k(w) = k(v)+1$, and $B = [\ ]$.)  Every vertex on the cycle
  that is not $w$ and not in $B$ must have level less than $k(w)$, and
  the forward search from $w$ will eventually visit each such vertex,
  traversing the cycle forward, until traversing an arc $(x,y)$ with
  $y=v$ or $y$ in $B$ and reporting a cycle.  We conclude that the
  algorithm reports a cycle if and only if an arc insertion creates
  one.  

  Suppose the insertion of $(v,w)$ does not create a cycle.  After the
  backward search stops, $B$ contains the vertices from which $v$ is
  reachable by a path all of whose vertices have level $k(w)$.  Also,
  if $(x,y)$ is an arc with $x$ not in $B$, but $y$ in $B$, $k(x) <
  k(y)$.  Step 3 increases to $k(w)$ the level of every vertex in
  $F$.  After the re-indexing in Step 4, consider any arc $(x,y)$.  If
  neither $x$ nor $y$ is in $B \& F$, $k(x), i(x) < k(y),i(y)$, since
  this was true before the insertion.  If both $x$ and $y$ are in $B \&
  F$, $k(x),i(x) < k(y),i(y)$ because $B \& F$ is in topological
  order.  If $y$ is in $B$ but $x$ is not in $B \& F$, then $k(x) <
  k(y)$.  If $y$ is in $F$ but $x$ is not in $B \& F$, then $k(x) <
  k(y)$ since $y$ increased in level but $x$ did not, and $k(x) \leq
  k(y)$ before the arc insertion.  If $x$ is in $B \& F$ but $y$ is not
  in $B\&F$, then $k(x) \leq k(y)$ and $i(x) < i(y)$ after the
  insertion.  We conclude that after the insertion, lexicographic
  order on level and index is topological
\end{proof}

\begin{lemma}
  The algorithm assigns no index less than $-nm - n$.
  \lemlabel{minindex}
\end{lemma}
\begin{proof}
  All initial indices are at least $-n$.  Each arc insertion decreases
  the minimum index by at most $n$, so after $m$ insertions the minimum
  index is at least $-nm - n$.
\end{proof}

\begin{lemma}
  No vertex level exceeds $\min\set{m^{1/2},n^{2/3}} + 2$.  
  \lemlabel{arcmaxlevel}
\end{lemma}
\begin{proof}
  Fix a topological order just before the last arc insertion.  Let $k
  > 1$ be a level assigned before the last arc insertion, and let $w$
  be the lowest vertex in the fixed topological order assigned level
  $k$.  For $w$ to be assigned level $k$, the insertion of an arc $(v, w)$
  must cause a backward search from $v$ that traverses at least $\Delta$
  arcs both ends of which are on level $k - 1$.  All the ends of these
  arcs must still be on level $k - 1$ just before the last insertion.
  Thus these sets of arcs are distinct for each $k$, as are their sets
  of ends.  Since there are only $m$ arcs, there are most $m/\Delta$
  distinct values of $k$.  Also, for each $k$ there must be at least
  $\sqrt{\Delta}$ distinct arc ends, since there are no loops or multiple
  arcs.  Since there are only $n$ vertices, there are at most
  $n/\sqrt{\Delta}$ distinct values of $k$.  It follows that no vertex level
  exceeds $\min\set{m/\Delta, n/\sqrt{\Delta}} + 2$, which gives the lemma.
\end{proof}

\begin{theorem}
  The insertion algorithm takes $O(\min\set{m^{1/2},n^{2/3}}m)$ time
  for $m$ arc insertions. 
  \thmlabel{sparsetime}
\end{theorem}
\begin{proof}
  By \lemreftwo{minindex}{arcmaxlevel}, all levels and indices are
  polynomial in $n$, so assignments and comparisons of levels and
  indices take $O(1)$ time.  Each backward search takes $O(\Delta) =
  O(\min\set{m^{1/2},n^{2/3}})$ time.  The time spent adding and
  removing arcs from incidence sets is $O(1)$ per arc added or
  removed.  An arc can be added or removed only when it is inserted
  into the graph or when the level of one of its ends increases.  By
  \lemref{arcmaxlevel}, this can happen at most
  $O(\min\set{m^{1/2},n^{2/3}})$ times per arc.  The time for a
  forward search is $O(1)$ plus $O(1)$ per arc $(x,y)$ such that $x$
  increases in level as the result of the arc insertion that triggers
  the search.  By \lemref{arcmaxlevel}, this happens
  $O(\min\set{m^{1/2},n^{2/3}})$ times per arc.
\end{proof}

\ 

The space needed by the algorithm is $O(m)$.

\begin{theorem}
  For any $n$ and $m$ with $m \leq n(n-1)/2$, there exists a
  sequence of $m$ arc insertions causing the algorithm to run
  in $\Omega(\min\set{m^{1/2},n^{2/3}}m)$ total time.
\end{theorem}

\begin{proof}
  Assume without loss of generality that $m \geq 2n$ and $n$ is
  sufficiently large.  Let the vertices be $1$ through $n$, numbered
  in the initial topological order.  We first add arcs consistent with
  the initial order (so that no reordering takes place) to construct a
  number of cliques of consecutive vertices.  An \defn{$r$-clique} of
  vertices $k$ through $k+r-1$ is formed by adding arc $(i,j)$
  for $i, j$ such that $k \leq i < j \leq k+r-1$.  An $r$-clique
  consists of $r$ vertices and $r(r-1)/2$ arcs.

  Let $r_1 = \floor{\sqrt{m}/2}$.  Construct an $r_1$-clique of the
  first $r_1$ vertices.  This is the \defn{main clique}.  The main clique
  contains at most $n/2$ vertices and at most $m/4$ arcs.  Let $r_2 =
  \ceil{\sqrt{\Delta}+1}$.  Starting with vertex $r_1+1$,
  construct $r_2$-cliques on disjoint sets of consecutive vertices,
  until running out of vertices or until $\floor{m/2}$ arcs have been
  added, including those added to make the main clique.  Each of the
  $r_2$-cliques is an \defn{anchor clique}.  The number of arcs in each
  anchor clique is $O(\Delta)$ and at least $\Delta$.  Number the anchor
  cliques from $1$ though $k$ in increasing topological order.  Then $k =
  \Theta(\Delta)$.  So far there has been no vertex reordering, and all
  vertices have level 1.

  Next, for $j$ from $k-1$ through $1$ in decreasing order, add arcs
  from the last vertex of anchor clique $j+1$ to each vertex of anchor
  clique $j$.  Add these arcs in decreasing topological order with
  respect to the end of the arc that is in anchor clique $j$.  There
  are at most $n/2 \leq m/4$ such arc additions.  Each addition of an
  arc from the last vertex of anchor clique $k$ to a vertex w in
  anchor clique $k-1$ triggers a backward search that traverses at
  least $\Delta$ arcs and causes the level of $w$ to increase from $1$
  to $2$.  Each forward search visits only a single vertex.  Once all
  arcs from anchor clique $k$ are added, all vertices in anchor clique
  $k-1$ have level $2$.  Addition of the arcs from the last vertex of
  anchor clique $k-1$ to the vertices in anchor clique $k-2$ moves all
  vertices in anchor clique $k-2$ to level $3$.  After all the arcs
  between anchor cliques are added, every vertex in anchor clique $j$
  is on level $k+1-j$.  The number of arcs added to obtain these level
  increases is at most $n/2 \leq m/4$.

  Finally, for each anchor clique from $k-2$ through $1$ in decreasing
  order, add an arc from its first vertex in topological order to the
  first vertex in the main clique.  There are at most $n/2 \leq m/4$ such
  arc additions.  Each addition triggers a backward search that visits
  only one vertex, followed by a forward search that traverses all the
  arcs in the main clique and increases the level of all vertices in
  the main clique by one.  These forward searches do $\Theta(\Delta m)$
  arc traversals altogether.  At most $m$ arcs are added during the
  entire construction.
\end{proof}

\secput{dense}{A One-Way-Search Algorithm for Dense Graphs}

The two-way-search algorithm becomes less and less efficient as the
graph density increases; for sufficiently dense graphs, one-way-search
is better.  In this section we present a one-way search algorithm that
takes $O(n^2\log{n})$ time for all arc insertions.   The
algorithm maintains for each vertex $v$ a level $k(v)$ that is a weak
topological numbering satisfying $k(v) \leq \id{size}(v)$.  The
algorithm pays for its searches by increasing vertex levels, using the
following lemma to maintain $k(v) \leq \id{size}(v)$ for all $v$.

\begin{lemma}
  In an acyclic graph, if a vertex $v$ has $j$ predecessors, each of size
  at least $s$, then $\id{size}(v) \geq s + j$.
  \lemlabel{size}
\end{lemma}
\begin{proof}
  Order the vertices of the graph in topological order and let $u$ be
  the smallest predecessor of $v$.  Then $\id{size}(v) \geq
  \id{size}(u) + j \geq s + j$.  Here ``$+j$'' counts $v$ and the
  $j-1$ predecessors of $v$ other than $u$.  
\end{proof}

\ 

The algorithm uses \lemref{size} on a hierarchy of scales.  For each
vertex $v$, in addition to a level $k(v)$, it maintains a bound
$b_i(v)$ and a count $c_i(v)$ for each integer $i, 0 \leq i \leq
\floor{\lg n}$, where $\lg$ is the base-2 logarithm.  Initially $k(v)
= 1$ for all $v$, and $b_i(v) = c_i(v) = 0$ for all $v$ and $i$.  To
represent the graph, for each vertex $v$ the algorithm stores the set
of outgoing arcs $(v, w)$ in a heap (priority queue) $\id{out}(v)$,
each arc having a \defn{priority} that is at most $k(w)$.  (This
priority is either $k(w)$ or a previous value of $k(w)$.)  Initially
all such heaps are empty.

The arc insertion algorithm maintains a set of arcs $A$ to be traversed,
initially empty.  To insert an arc $(v, w)$, add $(v, w)$ to $A$ and repeat
the following step until a cycle is detected or $A$ is empty:

\ 

\noindent\textbf{Traversal Step:}\vspace{-1em}
\begin{codebox}
  \li delete some arc $(x, y)$ from $A$
  \li \If $y = v$
  \li \Then stop the algorithm and report a cycle
      \End
  \li \If $k(x) \geq k(y)$
  \li \Then $k(y) \gets k(x) + 1$
  \li \Else \Comment{$k(x) < k(y)$}
  \li       $i \gets \floor{\lg(k(y) - k(x))}$
  \li       $c_i(y) \gets c_i(y)+1$
  \li       \If $c_i(y) = 3\cdot 2^{i+1}$
  \li       \Then $c_i(y) \gets 0$
  \li             $k(y) \gets \max\set{k(y), b_i(y) + 3\cdot 2^i}$
  \li             $b_i(y) \gets k(y) - 2^{i+1}$.  
            \End
      \End
  \li delete from $\id{out}(y)$ every arc
with priority at most $k(y)$ and add these arcs to $A$.  
  \li add $(x, y)$ to $out(x)$ with priority $k(y)$.
\end{codebox}

In a traversal step, an arc $(y, z)$ that is deleted from
$\id{out}(y)$ may have $k(z) > k(y)$, because $k(z)$ may have
increased since $(y, z)$ was last inserted into $\id{out}(y)$.
Subsequent traversal of such an arc may not increase k(z).  It is to
pay for such traversals that we need the mechanism of bounds and
counts.

We implement each heap $\id{out}(v)$ as an array of buckets indexed
from $1$ through $n$, with bucket $i$ containing the arcs with
priority~$i$.  We also maintain the smallest index of a nonempty
bucket in the heap.  This index never decreases, so the total time to
increment it over all deletions from the heap is $O(n)$.  The time to
insert an arc into a heap is $O(1)$.  The time to delete a set of arcs
from a bucket is $O(1)$ per arc deleted.  The time for heap operations
is thus $O(1)$ per arc traversal plus $O(n)$ per heap.  Since there
are $n$ heaps, this time totals $O(1)$ per arc traversal plus
$O(n^2)$.  

To analyze the algorithm, we begin by bounding the total number of arc
traversals, thereby showing that the algorithm terminates.  Then we
prove its correctness.  Finally, we fill in one detail of the
implementation and bound the running time.

\begin{lemma} While the graph remains acyclic, the insertion algorithm
  maintains $k(v) \leq \id{size}(v)$ for every vertex $v$.
  \lemlabel{densemaxlevel}
\end{lemma}
\begin{proof}
  The proof is by induction on the number of arc insertions.  The
  inequality holds initially.  Suppose it holds just before the
  insertion of an arc $(v,w)$ that does not create a cycle.  Consider
  a traversal step during the insertion that deletes $(x, y)$ from $A$
  and increases $k(y)$.  If $k(y)$ increases to $k(x) + 1$,
  $\id{size}(y) \geq 1 + \id{size}(x) \geq 1 + k(x)$, maintaining the
  inequality for $y$.  The more interesting case is when $c_i(y) =
  3\cdot 2^{i + 1}$ and k(y) increases to $b_i(y) + 3\cdot 2^i$.  Each
  of the increases to $c_i(y)$ since it was last zero corresponds to the
  traversal of an arc $(z, y)$.  When $c_i(y)$ was last zero, $b_i(y)
  = \max\set{0, k(y) - 2^{i + 1}}$.  Since $k(y)$ cannot decrease,
  $b_i(y) \leq k(z) \leq \id{size}(z)$ when this traversal of $(z, y)$
  occurs, since at this time $k(y) - k(z) < \min\set{k(y), 2^{i +
      1}}$.  We consider two cases.  If there were at least $3\cdot
  2^i$ traversals of distinct arcs $(z, y)$ since $c_i(y)$ was last
  zero, then $\id{size}(y) \geq b_i(y) + 3 \cdot 2^i$ by
  \lemref{size}, and the increase in $k(y)$ maintains the inequality
  for $y$.  If not, by the pigeonhole principle there were at least
  three traversals of a single arc $(z, y)$ since $c_i(y)$ was last
  zero.  When each traversal happens, $k(y) - k(z) \geq 2^i$, but each
  of the second and third traversals cannot happen until $k(z)$
  increases to at least the value of $k(y)$ at the previous traversal.
  This implies that when the third traversal happens, $k(y) \geq
  b_i(y) + 3\cdot 2^i$, so $k(y)$ will not in fact increase as a
  result of this traversal.
\end{proof}

\begin{lemma}
  If a new arc $(v, w)$ creates a cycle, the insertion algorithm
  maintains $k(v) \leq \id{size}(v) + n$, where sizes are
  before the addition of $(v, w)$.  
  \lemlabel{densecycle}
\end{lemma}
\begin{proof}
  Before the addition of $(v, w)$, $k(v) \leq \id{size}(v)$ for every
  vertex $v$, by \lemref{densemaxlevel}.  Traversal of the arc $(v,
  w)$ can increase $k(v)$ by at most $n$, so the desired inequality
  holds after this traversal.  Every subsequent traversal is of an arc
  other than $(v, w)$: to traverse $(v, w)$, an arc into $v$ must be
  traversed, which results in reporting of a cycle.  Thus the
  subsequent traversals are of arcs in the acyclic graph before the
  addition of $(v, w)$.  The proof of \lemref{densemaxlevel} extends
  to prove that these traversals maintain the desired inequality:
  \lemref{size} holds if the size function is replaced by the size
  plus any constant, in particular by the size plus $n$.
\end{proof}

\begin{lemma}
  The total number of arc traversals over $m$ arc additions is
  $O(n^2\log n)$.  \lemlabel{densetraversals}
\end{lemma}
\begin{proof}
  By \lemreftwo{densemaxlevel}{densecycle}, every label $k(v)$, and
  hence every bound $b_i(v)$, remains below $2n$.  Every arc traversal
  increases a vertex level or increases a count.  The number of level
  increases is $O(n^2)$.  Consider a count $c_i(v)$.  Each time
  $c_i(v)$ is reset to zero from $3\cdot 2^{i + 1}$, $b_i(v)$
  increases by at least $2^i$.  Since $b_i(v) \leq 2n$, the total
  amount by which $c_i(v)$ can decrease as a result of being reset is
  at most $12n$.  Since $c_i(v)$ starts at zero and cannot exceed
  $4n$, the total number of times $c_i(v)$ increases is at most $16n$.
  Summing over all counts for all vertices gives a bound of
  $O(n^2\log{n})$ on the number of count increases and hence on the
  number of arc traversals.
\end{proof}

\begin{theorem}
  If the insertion of an arc $(v,w)$ creates a cycle, the insertion
  algorithm stops and reports a cycle.  If not, the insertion
  algorithm maintains the invariant that $k$ is a weak topological
  numbering.  \thmlabel{densecorrect}
\end{theorem}
\begin{proof}
  By \lemref{densetraversals} the algorithm terminates.  A
  straightforward induction shows that every arc $(x, y)$ traversed by
  the insertion algorithm is such that $x$ is reachable from $v$, so
  if the algorithm stops and reports a cycle, there is one.  Suppose
  the insertion of $(v, w)$ creates a cycle.  Before the insertion of
  $(v, w)$, $k$ is a weak topological numbering, so the path from $w$
  to $v$ existing before the addition of $(v, w)$ has vertices in
  strictly increasing order.  Thus $v$ has the largest level on the
  path.  A straightforward induction shows that the algorithm will
  eventually traverse every arc on the path and report a cycle, unless
  it reports another cycle first.

  Suppose addition of an arc $(v, w)$ does not create a cycle.  Before
  the addition, $k$ is a weak topological numbering.  The algorithm
  maintains the invariant that every arc $(x, y)$ such that $k(x) \geq
  k(y)$ is either on $A$ or is the arc being processed.  Thus, once
  $A$ is empty, $k$ is a weak topological numbering.
\end{proof}

\begin{theorem}
  The algorithm runs in $O(n^2\log{n})$ total time.
  \thmlabel{densetime}
\end{theorem}
\begin{proof}
  The running time is $O(1)$ per arc traversal plus $O(n^2)$.  This is
  $O(n^2\log{n})$ by \lemref{densetraversals}.
\end{proof}

\ 

The space needed by the algorithm is $O(nlogn + m)$ for the labels,
bounds, and counts, and $O(n^2)$ for the n heaps.  Storing the heaps
in hash tables reduces their total space to $O(m)$ but makes the
algorithm randomized.  By using a two-level data
structure~\cite{TarjanYa79} to store each heap, the space for the
heaps can be reduced to $O(n^{1.5} + m)$ without using randomization.
This bound is $O(m)$ if $m/n = \Omega(n^{1/2})$; if not, the sparse
algorithm of Section 2 is faster.

The following theorem states that our analysis for this algorithm is
tight.  

\begin{theorem}
  For any sufficiently large $n$, there exists a sequence of
  $\Theta(n^2)$ arc insertions that causes our algorithm to do
  $\Omega(n^2\log{n})$ arc traversals.
  \label{thm:dense-lower}
\end{theorem}
\begin{proof}
  Without loss of generality, suppose $n = (7/2)r - 3$, where $r \geq
  2^3$ is a power of~$2$.  The graph we construct consists of three
  categories of vertices: (1) vertices $u_1,u_2,\ldots,u_r$, (2) sets
  of vertices $S_0,S_1,\ldots,S_{\lg(r) - 2}$ with $\card{S_j} = 3
  \cdot 2^{j+1}$ (so $\sum_j \card{S_j} = 3 (r/2 - 1)$), and (3) a set
  of vertices $T$ with $\card{T} = r$.  Initially there are no arcs
  in the graph, and all levels are~$1$.

  First, add arcs $(u_i,u_{i+1})$ in order for $1 \leq i < r$.  After
  these arc additions, $k(u_i) = i$.  These levels are invariant over
  the remainder of the arc insertions --- we use these vertices as
  anchors to increase the levels of all the other vertices.  In fact,
  the \emph{only} time the level of any other vertex $v \in (\bigcup_j
  S_j)\cup T$ will increase is when adding an arc $(u_i,v)$.
  
  The arc insertions proceed in phases ranging from $2$ to $r$.  In
  phase $i$, first insert arc $(u_{i-1},t)$ for all $t \in T$, thereby
  increasing $k(t)$ to $k(t) = i$.  Next, consider each $j$ for which
  there exists a constant $c \geq 3$ such that $i = c2^j$, i.e., $i$
  is a sufficiently large multiple of $2^j$.  There are two cases
  here, described in more detail shortly.  If $c=3$, insert arcs from
  $S_j$ to $T$, not causing a level increase to $t$.  If $c > 3$, the
  algorithm traverses the arcs from $S_j$ to $T$ again, but without
  causing any level increases to $t\in T$.  Moreover, the only time
  any $c_j(t)$ or $b_j(t)$ changes, for $j > 0$, is when the algorithm
  traverses an arc from $S_j$ to $t \in T$. 

  \noindent\textbf{Case 1 (add arcs from $S_j$ to $T$):} If $i = 3 \cdot
  2^j$ for some $j$, add arcs $(u_{2^{j+1}-1},s_j)$ for all $s_j \in
  S_j$, causing $k(s_j)$ to increase to $2^{j+1}$.  Also add arcs
  $(s_j,t)$ for all $s_j \in S_j$ and $t \in T$.  Observe that before
  these arc additions $\floor{\lg(k(t)-k(s_j))} =
  \floor{\lg(2^{j+1}-2^j)} = j$.  Moreover, $c_j(t) = 0$ and
  $b_j(t)=0$.  For each $t$, when the last arc insertion occurs,
  $c_j(t)$ increases to $3 \cdot 2^{j+1}$.  We have, however, that
  $k(t) = 3\cdot 2^{j+1} > b_j + 3\cdot 2^j$, and hence $k(t)$ does
  not increase.  The counter $c_j(t)$ is subsequently reset to $0$ and
  $b_j(t) = k(t) - 2^{j+1} = k(s_j) - 2^j$.  Finally, the priority of
  each of these arcs $(s_j,t)$ is updated to $3\cdot 2^j$ in
  $\id{out}(s_j)$.
  
  \noindent\textbf{Case 2 (follow arcs from $S_j$ to $T$):} Otherwise, $i =
  c2^j$, for $c > 3$.  Since $i > 3 \cdot 2^j$, the arcs $(s_j,t)$
  already exist.  Before this step, we have $k(s_j) = k(t) - 2^{j+1}$,
  for each $s_j \in S_j$.  Moreover, we have $b_j(t) = k(s_j) - 2^j =
  k(t) - 3\cdot 2^j$.  Insert arcs $(u_{i-2^j-1},s_j)$, for all $s_j
  \in S_j$.  Such an arc insertion causes $k(s_j)$ to increase to the
  next multiple of $2^j$.  After the update, we have $k(s_j)$ equal to
  the priority of each arc $(s_j,t)$ in $\id{out}(s_j)$, and hence the
  algorithm traverses each of the outgoing arcs.  Moreover, $\lg(k(t)
  - k(s_j)) = \lg(i - 2^j) = j$, and hence the counter $c_j$ is
  affected.  For each $t$, the counter $c_j(t)$ again reaches $3\cdot
  2^{j+1}$.  Since $b_j(t) = k(t) - 3\cdot 2^j$, the level of $t$
  again does not increase.  The counter $c_j(t)$ is subsequently reset
  to $0$, each $b_j(t) = k(t) - 2^{j+1} = k(s_j) - 2^j$, and the
  priority of each of the arcs $(s_j,t)$ is set to $k(t)$ in
  $\id{out}(s_j)$.
  
  In both cases, whenever the phase number $i$ is a large enough
  multiple of $2^j$, the algorithm traverses all arcs $(s_j,t)$ such
  that $s_j \in S_j$ and $t \in T$.  Consider a fixed $j$.  There are
  $\card{S_j} \cdot \card{T} = 3\cdot 2^{j+1}r$ such arcs.  Summing
  over all $r/2^j-2$ phases during which the phase number is a large
  enough multiple of $2^j$, there are $(3\cdot 2^{j+1}r)(r/2^j-2) =
  \Omega(r^2) = \Omega(n^2)$ arc traversals from vertices in $S_j$ to
  vertices in $T$.  Summing over all $\lg(r)-2 = \Theta(\log n)$
  values of $j$ yields a total of $\Omega(n^2 \log n)$ arc traversals.
\end{proof}

\ 

The proof extends to give a slightly more
general result: for any $1\leq k \leq \lg n$, there is a sequence of
$\Theta(2^k n)$ arc insertions causing the algorithm to do $\Theta(n^2
k)$ arc traversals.  To prove this, omit from the proof of
\thmref{dense-lower} the sets $S_j$ with $j>k$.  The generalization
implies that $\Theta(n)$ arcs are enough to make the algorithm take
$\Omega(n^2)$ time, and $\Theta(n^{1+\epsilon})$ arcs, for any
constant $\epsilon > 0$, are enough to make the algorithm take
$\Omega(n^2 \log n)$ time. 

\secput{extensions}{Simple Extensions}

In this section we extend our sparse and dense algorithms to provide
some additional capabilities possessed by previous algorithms.  All
the extensions are simple and preserve the asymptotic time bounds of
the unextended algorithms.  Our first extension eliminates ties in the
vertex numbering maintained by the dense algorithm presented in
\secref{dense}.  We break ties by giving each vertex a distinct index
as in the sparse algorithm and ordering the vertices lexicographically
by level and index.  The indices can be arbitrary, as long as they are
distinct within each level: we can use fixed indices, or we can assign
new indices when vertices change level.

Assigning new indices is useful in our second extension, which
explicitly maintains a doubly-linked list of the vertices in
lexicographic order by level and index, and hence in a topological
order.  We maintain a pointer to the first vertex on the list.  We
also maintain for each non-empty level a pointer to the last vertex on
the level.  We store these pointers in an array indexed by level.
When a vertex increases in level, we delete it from its current list
position and re-insert it after the last vertex on its old level,
unless it was the last vertex on its old level, in which case its
position in the list does not change.  This takes O(1) time, including
all needed pointer updates.  In the sparse algorithm, when moving a
group of vertices whose levels change as a result of an arc insertion,
we move them in decreasing order by new index.  In the dense
algorithm, we can move such a group of vertices in arbitrary order,
but we then assign each vertex moved a new index that is less than those of
vertices previously on the level.  As in the sparse algorithm, we can
do this by maintaining the smallest index and counting down.

Our third extension explicitly returns a cycle when one is discovered,
rather than just reporting that one exists.  We augment each search to
grow a spanning tree represented by parent pointers as each search
proceeds.  In the sparse algorithm, the backward search generates an
in-tree rooted at $v$ containing all visited vertices; the forward
search generates an out-tree rooted at $w$ containing all vertices whose
level increases.  If the backward search causes $k(w)$ to increase to
$k(v) + 1$ and $B$ to become empty, the forward search may visit vertices
previously visited by the backward search.  Each such vertex acquires
a new parent when the forward search visits it for the first time.
When the algorithm stops and reports a cycle, a cycle can be obtained
explicitly by following parent pointers.  Specifically, if the
backward search traverses an arc $(w, y)$, following parent pointers
from $y$ gives a path from $y$ to $v$, which forms a cycle with $(v, w)$ and
$(w, y)$.  If the forward search traverses an arc $(x, y)$ with $y = v$ or $y$
in $B$, traversing parent pointers from $x$ and from $y$ gives a path from $w$
to $x$ and a path from $y$ to $v$, which form a cycle with $(x, y)$ and $(v,
w)$.  In the dense algorithm, there is only one tree, an out-tree
rooted at $v$, containing $v$ and all vertices whose level increases.
Vertex $v$ has one child, $w$.  If the search traverses an arc $(x, v)$,
following parent pointers from $x$ gives a path from $v$ through $(v, w)$ to
$x$, which forms a cycle with $(x, v)$.

Our fourth extension is to handle vertex insertions and to allow $n$
and $m$ to be unknown in advance.  In the sparse algorithm, we insert
a vertex $v$ by giving $v$ a level of 1, decrementing $\id{index}$,
and giving $v$ an index equal to $\id{index}$.  We also maintain a
running count of $n$ and $m$.  Each time $n$ or $m$ doubles, we
recompute $\Delta$, but we replace $\Delta$ only if it doubles.  It is
straightforward to verify that \thmref{sparsetime} remains true.  In
the dense algorithm, we insert a vertex by giving it a level of 1.  We
also maintain a running count of $n$.  Each time $\floor{lg{n}}$
increases, we add a corresponding new set of bounds and counts.
\thmref{densetime} remains true.  We can combine this extension with
any of the other extensions in this section, and with the extension
described in the next section.

\secput{components}{Maintenance of Strong Components}
\newcommand{\find}{\proc{Find}}
\newcommand{\link}{\proc{Link}}

A less straightforward extension of our algorithms is to the
maintenance of strong components.  This has been done for some of the
earlier algorithms by previous authors.  Pearce~\cite{Pearce05} and
Pearce and Kelly~\cite{PearceKe03} sketched how to extend their
incremental topological ordering algorithm and that of
Marchetti-Spaccamela et al.~\cite{Marchetti-SpaccamelaNaRo96} to
maintain strong components; HKMST showed in detail how to extend their
algorithms.  Our strong component extensions follow the approach of
HKMST.

Our extended algorithms store the vertex sets of the strong components
in a disjoint set data structure~\cite{Tarjan75}.  Such a data structure
represents a partition of a set into disjoint subsets.  Each subset
has a \defn{canonical element} that names the set.  Initially all
subsets are singletons; the unique element in each subset is its
canonical element.  The data structure supports two operations:

$\find(x)$: Return the canonical element of the set containing
element $x$.

$\link(x, y)$: Form the union of the sets whose canonical
elements are $x$ and $y$, with the union having canonical element $x$.
This operation destroys the old sets containing $x$ and $y$.

We store each set as a tree whose nodes are the elements of the set,
with each element having a pointer to its parent.  If we do
\proc{Find}s using path compression, and we do \proc{Link}ing by rank
or size, then the total time for any number of \proc{Find} and
\proc{Link} operations on a partition of $n$ elements is $O(n\log n)$
plus $O(1)$ per operation~\cite{Tarjan75}.  In our application the number
of set operations is $O(1)$ per arc examined, so the time for the set
operations does not increase the asymptotic time bound.

\subsection{An Extension for Sparse Graphs}

Our extension of the sparse algorithm maintains for each component a
level, an index, a set of arcs $(x, y)$ such that $x$ is in the
component, and a set of arcs $(x, y)$ such that $y$ is in the
component and the components containing $x$ and$ y$ are on the same
level.  We store the level, index, and incident arc sets with the
canonical vertex of the component.  Initially every vertex is in its
own component, all components are on level 1, the components have
distinct indices between $-1$ and $-n$, inclusive, all the incident
arc lists are empty, and \id{index}, the smallest index, is $-n$.  Each
search generates a spanning tree, represented by parent pointers: if
$x$ is a canonical vertex, $p(x)$ is its parent; if $x$ is a root,
$p(x) = x$.  Initially all parents are \const{null}.  The algorithm
marks canonical vertices it finds to be in a new component.  Initially
all vertices are unmarked.

The algorithm adds a new arc to the appropriate incident arc lists
only if its ends are in different components after its insertion.  The
backward searches delete arcs whose ends are in the same component, as
well as the second and subsequent arcs between the same pair of
components.  To facilitate the latter deletions, it uses a bit matrix
$M$ indexed by vertex.  Initially all entries of $M$ are zero.

The algorithm for inserting a new arc $(v, w)$ consists of the following
steps:

\noindent \textbf{Step 1 (test order):} Let $u = \find(v)$ and $z =
\find(w)$.  If $k(u), i(u) \leq k(z), i(z)$, go to Step 6.

\noindent \textbf{Step 2 (search backward):} Let $B = [ ]$,$F = [ ]$,
$\id{arcs} = 0$, $p(u) = u$, and $p(z) = z$.  Do $\proc{EBvisit}(u)$,
where mutually recursive procedures \proc{EBvisit} and
\proc{EBtraverse} are defined as follows:

$\proc{EBvisit}(t)$: For $(x, y) \in \id{in}(t)$ do $\proc{EBtraverse}(x, y)$.  Let $B = B \& [t]$.

$\proc{EBtraverse}(x, y)$: If $\find(x) = \find(y)$, or $M(\find(x),
\find(y)) = 1$, delete $(x, y)$ from $\id{in}(\find(y))$ and from
$\id{out}(\find(x))$.  Otherwise, do the following.  Let $M(\find(x),
\find(y)) = 1$ and $\id{arcs} = \id{arcs} + 1$.  If $\id{arcs} \geq
\Delta$, let $k(z) = k(u) + 1$, let $\id{in}(z) = \set{}$, let $B = [
]$, unmark any canonical vertices marked as being in a new component,
make all parents null except that of $z$, reset $M$ to zero, and go to
Step 3.  If $\find(x) = z$, mark $z$ as being in a new component.  If
$\find(x)$ is marked, follow parent pointers from $\find(y)$, marking
every canonical vertex reached (including $\find(y)$) as being in a
new component, until marking $u$ or reaching a marked vertex.  If
$p(\find(x)) = \const{null}$, let $p(\find(x)) = \find(y)$ and do
$\proc{EBvisit}(\find(x))$.

If the search ends before traversing at least $\Delta$ arcs, test whether
$k(z) = k(u)$.  If so, go to Step 4; if not, let $k(z) = k(u)$ and $\id{in}(z) =
\set{}$.

\noindent \textbf{Step 3 (search forward)}: Do $\proc{EFvisit}(z)$,
where mutually recursive procedures \proc{EFvisit} and
\proc{EFtraverse} are defined as follows:

$\proc{EFvisit}(t)$: For $(x, y)$ in $\id{out(t)}$ do $\proc{EFtraverse}(x, y)$.  Let $F = [t] \& F$.

$\proc{EFtraverse}(x, y)$: If $\find(x) = \find(y)$, delete $(x, y)$
from $\id{out}(\find(x))$.  If $\find(y) = u$ or $\find(y)$ is on $B$,
follow parent pointers from $\find(y)$, marking each canonical vertex
reached (including $\find(y)$), until marking $u$ or reaching a marked
vertex.  If $\find(y)$ is marked, follow parent pointers from
$\find(x)$, marking each canonical vertex reached (including
$\find(x)$), until marking $z$ or reaching a marked vertex.  If
$k(\find(y)) < k(z)$, let $p(\find(y)) = \find(x)$, let $k(\find(y)) =
k(z)$, let $\id{in}(\find(y)) = \set{}$, and do
$\proc{EFvisit}(\find(y))$.  If $k(\find(y)) = k(z)$, add $(x, y)$ to
$\id{in}(\find(y))$.

\noindent \textbf{Step 4 (form component):} Let $L = B \& F$.  If $z$
is marked, combine the old components containing the marked vertices
into a single new component with canonical vertex $z$ by uniting the
incoming and outgoing arc sets of the marked vertices, and uniting the
vertex sets of the components using \proc{Unite}.  Delete from $L$ all
marked vertices.  Unmark all marked vertices.

\noindent \textbf{Step 5 (re-index):} While $L$ is non-empty, let
$\id{index} = \id{index} - 1$, delete the last vertex $x$ on $L$, and
let $i(x) = \id{index}$.

\noindent \textbf{Step 6 (add arc):} If $\find(v) \neq \find(w)$, add $(v, w)$ to $\id{out}(u)$ and, if $k(u) = k(z$), to $\id{in}(z)$.

In the proofs to follow we denote levels and indices just before and
just after the insertion of an arc $(v, w)$ by unprimed and primed
values, respectively.

\begin{theorem}
  The extended sparse algorithm is correct.  That is, it correctly
  maintains the strong components, all the data structures, and the
  following invariant on the levels and indices: if $(x, y)$ is an arc,
  either $\find(x) = \find(y)$ or $k(\find(x)), i(\find(x) < k(\find(y)),
  i(\find(y))$.
\end{theorem}
\begin{proof}
  The proof is by induction on the number of arc insertions.
  Initially all the data structures are correct.  It is
  straightforward to verify that the algorithm correctly maintains
  them, assuming that it correctly maintains the strong components and
  the desired invariant on levels and indices.  Suppose the strong
  components are correct and the invariant holds before the insertion
  of an arc $(v, w)$.  If this insertion does not create a new
  component, then the algorithm does the same thing as the unextended
  algorithm, except that it operates on components instead of
  vertices.  Thus after the insertion the components are correct and
  the invariant holds.

  Suppose on the other hand that the insertion of $(v, w)$ creates a new
  component.  Until a vertex is marked, the algorithm does the same
  thing as the unextended algorithm, except that it operates on
  components instead of vertices.  Thus it will mark at least one
  vertex.  We consider three cases: $k'(z) = k(z)$; $k'(z) = k(u) > k(z)$;
  $k'(z) = k(u) + 1 > k(z)$.

  If $k'(z) = k(z)$, then $k(z) = k(u)$, the insertion changes no levels,
  and the backward search finishes without traversing at least $\Delta$
  arcs.  An old canonical vertex $x$ is in the new component if and only
  it is on a simple path from $z$ to $u$.  All components along such a
  path must have level $k(u)$.  An induction shows that the backward
  search marks a canonical vertex if and only if it is on such a path.
  After Step 2, $B$ contains the canonical vertices on level $k(u)$ from
  which there is a path to $u$ avoiding the old component containing $z$.
  The algorithm skips the forward search.  It correctly forms the new
  component in Step 4 and deletes from $L = B$ all old canonical
  vertices in the new component except z.  This includes $u$.  The
  canonical vertices remaining in $L$ at the end of Step 4 are not
  reachable from $z$. It follows that Step 5 restores the invariant that
  if $(x, y)$ is an arc, $\find(x) = \find(y)$ or $k(\find(x)), i(\find(x)) <
  k(\find(y)), i(\find(y))$.

  A similar argument applies if $k'(z) = k(u) + 1 > k(z)$.  In this
  case $B = [ ]$ in Step 3.  At the end of Step 3, $F$ contains all
  old canonical vertices reachable from $z$ by a path through
  components whose old levels are at most $k(u)$.  This includes $u$.
  An old canonical vertex $x$ is in the new component if and only it
  is on a simple path from $z$ to $u$.  All components along such a
  path must have old level at most $k(u)$.  An induction shows that
  the forward search marks a canonical vertex if and only if it is on
  such a path.  It follows that the algorithm correctly forms the new
  component in Step 4 and restores the invariant on levels and indices
  in Step 5.

  The remaining case, $k'(z) = k(u) > k(z)$, is the most interesting:
  the backward search runs out of arcs to traverse, and there is a
  forward search.  After Step 2, $B$ contains all old canonical vertices
  from which $u$ is reachable by a path through components of old level
  $k(u)$.  After Step 3, $F$ contains all canonical vertices reachable
  from $z$ by a path through components with old levels less than $k(u)$.
  Thus $B$ and $F$ are disjoint.  An old canonical vertex $x$ is in the new
  component if and only if it is on a simple path from $z$ to $u$.  Such a
  path passes through components in non-increasing order by old level.
  These components have canonical vertices in $F$ or $B$, with those in $F$
  first.  An induction shows that Steps 2 and 3 mark a canonical
  vertex if and only if it is on such a path.  It follows that the
  algorithm correctly forms the new component in Step 4 and restores
  the invariant on levels and indices in Step 5. 
\end{proof}

\lemref{minindex} remains true for the extended algorithm.  To bound
the running time, we need to prove \lemref{arcmaxlevel} for the
extension.  This requires some definitions.  We call an arc $(x, y)$
\defn{live} if $x$ and $y$ are in different strong components and
\defn{dead} otherwise.  A newly inserted arc that forms a new
component is dead immediately.  The \defn{level} of a live arc $(x,
y)$ is $k(\find(x))$.  The level of a dead arc is its highest level
when it was live; an arc that was never live has no level.  We
identify each connected component with its vertex set; an arc
insertion either does not change the components or combines two or
more components into one.  A component is \defn{live} if it is a
component of the current graph and \defn{dead} otherwise.  The
\defn{level} of a live component is the level of its canonical vertex;
the level of a dead component is its highest level when it was live.
A vertex and a component are \defn{related} if there is a path
that contains the vertex and a vertex in the component.  The number of
components, live and dead, is at most $2n - 1$.

\begin{lemma}In the extended sparse algorithm, no vertex level exceeds
  $\min\set{m^{1/2}, 2n^{2/3}} + 1$.
  \lemlabel{sparsecomponentlabel}
\end{lemma}
\begin{proof}
  We claim that for any level $k > 1$ and any level $j < k$, any canonical
  vertex of level $k$ is related to at least $\Delta$ arcs of level~$j$ and
  at least $\sqrt{\Delta}$ components of level~$j$.  We prove the claim by
  induction on the number of arc insertions.  The claim holds
  vacuously before the first insertion.  Suppose it holds before the
  insertion of an arc $(v, w)$.  Let $u = \find(v)$ and $z = \find(w)$ before
  the insertion.  A vertex is reachable from $z$ after the insertion if
  and only if it is reachable from $z$ before the insertion.  The
  insertion increases the level only of $z$ and possibly of some
  vertices and components reachable from $z$.  It follows that the claim
  holds after the insertion for any canonical vertex not reachable
  from $z$.

  Consider a vertex $y$ that is reachable from $z$ and is canonical after
  the insertion.  Since level order is topological, $k'(y) \geq k'(z)$.
  For $j$ such that $k'(z) \leq j < k'(y)$, $y$ is related to at least $\Delta$
  arcs of level~$j$ and $\sqrt{\Delta}$ components of level~$j$ before the
  insertion.  None of these arcs or components changes level as a
  result of the insertion, so the claim holds after the insertion for
  $y$ and level~$j$.  Since any arc or component of level less than $k'(z)$
  that is related to $z$ is also related to $y$, the claim holds for $y$
  after the insertion if it holds for~$z$.

  After the insertion, $z$ is reachable from $u$.  Also, $k(u) \leq
  k'(z) \leq k(u) + 1$.  The claim holds for $u$ before the insertion.
  Let $(x, y)$ be an arc of level less than $k(u)$ that is related to
  $u$ before the insertion.  If $x$ is reachable from $z$, $(x, y)$
  will be dead after the insertion and hence its level will not
  change.  Neither does its level change if $x$ is not reachable from
  $z$.  Arc $(x, y)$ is related to $z$ after the insertion.  Consider
  a component of level less than $k(u)$ that is related to $u$ before
  the insertion.  If the component is reachable from $z$, it is dead
  after the insertion of $(v,w)$ and hence does not change level; if
  it is not reachable from $z$, it also does not change level.  After
  the insertion, the component is related to~$z$.  It follows that the
  claim holds for $z$ and any level $j < k(u)$.

  One case remains: $j = k(u) < k'(z) = k(u) + 1$.  For the level of
  $z$ to increase to $k(u) + 1$, the backward search must traverse at
  least $\Delta$ arcs of level $k(u)$ before the insertion, each of
  which is related to $z$ and on level $k(u)$ after the insertion.
  The ends of these arcs are in at least $\sqrt{\Delta}$ components of
  level $k(u)$, each of which is related to $z$ and on level $k(u)$
  after the insertion.  Thus the claim holds for $z$ and level $k(u)$
  after the insertion.  This completes the proof of the claim.

  The claim implies that for every level other than the maximum, there
  are at least $\Delta$ different arcs and $\sqrt{\Delta}$ different
  components.  Since there are only $m$ arcs and at most $2n - 1$
  components, the maximum level is at most $\min\set{m/\Delta,
    2n/\sqrt{\Delta}} + 1$.  The lemma follows.
\end{proof}

\begin{theorem} The extended sparse algorithm takes
  $O(\min\set{m^{1/2}, n^{2/3}} m)$ time for $m$ arc insertions.
\end{theorem}
\begin{proof}
  The proof is like the proof of \thmref{sparsetime}, using
  \lemref{sparsecomponentlabel}. 
\end{proof}

The space required by the extended algorithm is $O(n^2)$, since the
bit matrix $M$ requires $O(n^2)$ space (or less if bits are packed
into words).  If we store $M$ in a hash table, the space becomes
$O(m)$ but the algorithm becomes randomized.  By using a three-level
data structure~\cite{TarjanYa79} to store $M$ we can reduce the space
to $O(n^{4/3}+m)$ without using randomization.  We obtain a simpler
algorithm with a time bound of $O(m^{3/2})$ by eliminating the
deletion of multiple arcs, thus avoiding the need for $M$, and letting
$\Delta = m^{1/2}$.  If we run this simpler algorithm until $m >
n^{4/3}$, then start over with all vertices on level one and indexed in
topological order and run the more-complicated algorithm with $M$ stored
in a three-level data structure, we obtain a deterministic algorithm
running in $O(\min\set{m^{1/2}, n^{2/3}} m)$ time and $O(m)$ space.

\subsection{ An Extension for Dense Graphs}

Our extension of the dense algorithm does two searches per arc
addition, the first to find cycles, the second to update levels,
bounds, and counts.  The levels, bounds, counts, and arc heaps are of
components, not vertices.  We store these values with the canonical
vertices of the components.  Initially each vertex is in its own
component, all levels are one, all bounds and counts are zero, and all
heaps are empty.  The algorithm deletes arcs with both ends in the
same component, as well as the second and subsequent arcs between the
same pair of components.  As in the sparse extension, to do the latter
it uses a bit matrix $M$ indexed by vertex, initially identically zero.

To insert an arc $(v, w)$, let $u = \find(x)$ and $z = \find(y)$.  If $k(u) <
k(z)$, add $(v, w)$ to $\id{out}(u)$ with priority $k(z)$.  If $k(u)
\geq k(z)$ and $u \neq z$, do Steps 1-4 below.  (If $u = z$ do nothing.)

\noindent \textbf{Step 1 (search for cycles):}  Let $S = \set{(v, w)}$.  Mark $u$.  Repeat the following step until $S$ is empty:

\emph{Cycle Traversal:} Delete some arc $(x, y)$ from $S$ and add it to $A$.  If
$\find(y)$ is marked, follow parent pointers from $\find(x)$, marking each
canonical vertex reached, until reaching a previously marked vertex.
If $k(\find(y)) < k(u)$, let $k(\find(y)) = k(u)$, let $p(\find(y)) = \find(x)$,
and delete from $\id{out}(\find(y))$ all arcs with priority at most $k(u)$ and
add them to $S$.

\noindent \textbf{Step 2 (form component):} If $z$ is marked, unite
the components containing the marked canonical vertices into a single
new component whose canonical vertex is $z$.  Form the new arc heap of $z$
by melding the heaps of the marked vertices, including $z$.  Unmark all
marked vertices.

\noindent \textbf{Step 3 (update levels, bounds, and counts):} Repeat
the following step until $A$ is empty:

\emph{Update Traversal:} Delete some arc $(x,y)$ from $A$.  If $\find(x) \neq
\find(y)$ and $M(\find(x), \find(y)) = 0$, proceed as follows.  Let
$M(\find(x), \find(y)) = 1$.  If $k(\find(x)) \geq k(\find(y))$, increase
$k(\find(y))$ to $k(\find(x)) + 1$; otherwise, let $i = \floor{\lg(k(\find(y)) -
k(\find(x)))}$, add one to $c_i(\find(y))$, and if $c_i(\find(y)) =
3*2^{i + 1}$, set $c_i(\find(y)) = 0$, set $k(\find(y)) =
\max\set{k(\find(y)), b_i(\find(y)) + 3*2^i}$, and set $b_i(\find(y)) =
k(\find(y)) - 2^{i + 1}$.  Delete from $\id{out}(\find(y))$ each arc with
priority at most $k(\find(y))$ and add it to $A$.  Add $(x, y)$ to
$\id{out}(\find(x))$ with priority $k(\find(y))$.

\noindent \textbf{Step 4 (reset $M$):} Reset each 1 in M to 0.

\begin{theorem}
  The extended dense algorithm correctly maintains strong components
  and the inequality $k(v) \leq size(v)$ for every vertex $v$.
\end{theorem}
\begin{proof}
  The proof is by induction on the number of arc insertions.  The
  theorem holds initially.  Suppose it holds before the insertion of
  an arc $(v, w)$.  Let $u = \find(v)$ and $z = \find(w)$ before the
  insertion.  If $u = z$ or $k(u) < k(z)$, the theorem holds after the
  insertion.  Suppose $u \neq z$ and $k(u) \geq k(z)$, so that Steps
  1-4 are executed.  Step 1 visits all vertices of level less than
  $k(u)$ reachable from $z$ and increases their level to $k(u)$.
  Since $z$ is reachable from $u$ after the insertion of $(v, w)$, all
  such vertices have size at least $k(u)$ after the insertion.  Thus
  such increases in level maintain the inequality between levels and
  sizes.  If the insertion of $(v, w)$ creates a new component, the
  vertices in the component are exactly those on paths from $z$ to
  $u$, all of which must have level less than $k(u)$ before the
  insertion.  Thus Step 1 will visit all such vertices other than $u$
  and mark all of them, including $u$, and Step 2 will correctly form
  the new component.  Step 3 updates levels, bounds, and counts
  exactly as in the unextended algorithm except that it operates on
  components, not vertices, and it traverses no arcs with both ends in
  the same component and at most one arc between any pair of
  components.  The proof of \lemref{densemaxlevel} extends to show
  that Step 3 maintains the inequality between levels and sizes.  Thus
  the theorem holds after the insertion.
\end{proof}

\begin{theorem}
  The extended dense algorithm runs in $O(n^2\log n)$ total time.
\end{theorem}
\begin{proof}
  The proof of \lemref{densetraversals} extends to show that the
  extended dense algorithm does $O(n^2\log n)$ arc traversals, from which
  the theorem follows.
\end{proof}

\ 

The space required by the extended dense algorithm is $O(n^2)$, or
$O(n\log n + m)$ if the heaps and the matrix $M$ are stored in hash tables.

\secput{conc}{Concluding Remarks}

We have presented two algorithms for incremental cycle detection and
related problems, one for sparse graphs and one for dense graphs.
Their total running times are $O(\min\set{m^{1/2}, n^{2/3}}m)$ and
$O(n^2\log n)$, respectively.  The sparse algorithm is faster for graphs
whose density $m/n$ is $o(n^{1/3}\log n)$; the dense algorithm is faster for
graphs of density $\omega(n^{1/3}\log n)$.  The $O(n^{2/3}m)$ bound of the
sparse algorithm is best only for graphs with density in the sliver
from $\omega(n^{1/3})$ to $o(n^{1/3}\log n)$.  The HKMST paper gives a lower
bound of $\Omega(nm^{1/2})$ for algorithms that do vertex updates only
within the so-called ``affected region,'' the set of vertices that are
definitely out of order when a new arc is added.  Unlike previous
algorithms, our algorithms do not do updates completely within the
affected region, yet they do not beat the HKMST lower bound, and we
have no reason to believe it can be beaten.  On the other hand, for
graphs of intermediate density our bounds are far from $O(nm^{1/2})$,
and perhaps improvements coming closer to this bound are possible.

Another interesting research direction is to investigate whether batch
arc additions can be handled faster than single arc additions (other
than by reverting to a static algorithm if the batch is large enough).
See~\cite{AlpernHoRo90,PearceKe07}.  One may also ask whether arc
deletions, instead of or in addition to insertions, can be handled.
Our algorithms remain correct if arcs can be deleted as well as
inserted, but the time bounds are no longer valid, and we have no
interesting bounds.  Maintaining strong components as arcs are
deleted, or as arcs are inserted and deleted, is an even more
challenging problem.  See~\cite{RodittyZw08} and the references
contained therein.

\bibliographystyle{plain}
\bibliography{topsort}

\end{document}